\documentclass[aps,pra,twocolumn,superscriptaddress,nofootinbib]{revtex4-2}
\usepackage{amsmath,amssymb,mathtools,amsthm}
\usepackage{booktabs,graphicx}
\usepackage{xcolor}
\usepackage[colorlinks,linkcolor=blue!50!black,citecolor=blue!50!black,urlcolor=blue!50!black]{hyperref}

\newcommand{\vH}{v_{\mathrm{H}}}
\newcommand{\vx}{v_{\mathrm{x}}}
\newcommand{\Ex}{E_{\mathrm{x}}}
\newcommand{\rd}{\mathrm{d}}
\newcommand{\ee}{\mathrm{e}}
\newcommand{\Ax}{A_{\mathrm{x}}}
\newcommand{\Fmix}{F_{\mathrm{mix}}}
\newcommand{\Fpbe}{F_{\mathrm{PBE}}}

\theoremstyle{plain}
\newtheorem{theorem}{Theorem}
\newtheorem{proposition}[theorem]{Proposition}

\theoremstyle{definition}
\newtheorem{remark}[theorem]{Remark}

\begin{document}

\title{Semilocal exchange functionals from the exact-exchange condition for the hydrogen atom:\\
Hydrogenic exactness and recovery of Rydberg-like bound states}

\author{Fumihiro Imoto}

\date{July 2026}

\begin{abstract}
In semilocal density functionals the exchange potential decays too
rapidly outside atoms and molecules and lacks the correct $-1/r$ tail;
as a consequence, functionals from PBE to modern meta-GGAs such as SCAN
support no bound Rydberg-like state of an atom. We revisit the
gradient-corrected exchange functional that Gill and Pople (GP93)
constructed to reproduce the exact exchange potential of the hydrogen
atom, solve their equation as an inverse problem, and use
the resulting enhancement factor---the GP93 factor---as a
tail-generating ingredient of a switched semilocal functional. The GP93
factor is exact on the hydrogen $1s$ density and, by uniform coordinate
scaling, on hydrogenic $1s$ ions, where it yields the eigenvalue
$-Z^2/2$ exactly (hydrogenic exactness). Its large-gradient growth has
the hydrogen-exact form $O[s(\ln s)^{2/3}]$, distinct from the $s\ln s$
form of Armiento and K\"ummel. The factor is combined with a
kinetic-energy-density indicator so that its divergent branch acts only
in one-electron-like regions; the GP93 ingredient is nonempirical, while
a few switching parameters are calibrated to balance tail recovery
against self-consistent-field stability. The resulting functional
produces Rydberg-like series of bound virtual Kohn--Sham states in
systems whose outermost shell is a one-electron-like $s$
shell---all-electron H and He, and Li, Na, and K with large-core
pseudopotentials---and moves the highest-occupied eigenvalue toward the
experimental ionization energy. For $p$-shell atoms (Ne, Ar) the switch
closes, so the domain of applicability is fixed by design. Fixed-density tests indicate that
a reduced-Laplacian switch distinguishes atomic tails from covalent bond
centers, and that a density-only kinetic-energy functional renders the
entire switch orbital-independent, so that no generalized Kohn--Sham
solver is required.
\end{abstract}

\maketitle

\section{Introduction}\label{sec:intro}

The semilocal exchange-correlation functionals that power practical
Kohn--Sham~\cite{HK,KS} density-functional theory---GGAs such as
PBE~\cite{PBE} and meta-GGAs such as TPSS~\cite{TPSS} and
SCAN~\cite{SCAN}---have been remarkably successful for thermochemistry,
structure, and cohesive properties. Their exchange potential
$\vx(\mathbf r)$, however, decays exponentially in the outer region of a
finite system, where the density itself decays exponentially, and therefore
lacks the exact Coulombic tail
\begin{equation}
\vx(\mathbf r)\;\xrightarrow{\;r\to\infty\;}\;-\frac1r .
\label{eq:coulombtail}
\end{equation}
A well-known family of failures follows from this single structural
defect: unbound Rydberg states, unstable negative ions, systematically
underestimated ionization energies~\cite{PPLB,Janak}, and the excessive charge
delocalization caused by the self-interaction error (SIE). We begin with
the numerical fact that motivates this work: in the Kohn--Sham virtual
spectrum of the He atom, PBE, SCAN, TPSS, and M06-L support \emph{no}
bound Rydberg orbital at all (Sec.~\ref{sec:results},
Fig.~\ref{fig:ryd}). Even the latest meta-GGAs, despite the additional
variable $\tau$, do not restore Eq.~\eqref{eq:coulombtail}.

Existing attempts to repair the asymptotics fall into two broad classes.
Model potentials (LB94~\cite{LB94}, mBJ~\cite{mBJ}) construct the shape of
$\vx$ directly but possess no parent energy functional and hence lack
variational consistency. Range-separated hybrids~\cite{LC,CAMB3LYP} restore the long-range tail
through nonlocal exact exchange, at the corresponding computational cost.
Within the semilocal energy form, Armiento and K\"ummel
(AK13)~\cite{AK13} pursued improved asymptotics with an enhancement
factor of $s\ln s$ growth; their construction, however, is independent of
any exactness condition for hydrogen.

This paper takes a third route. In 1993 Gill and Pople derived the
ordinary differential equation that a GGA must satisfy in order to
reproduce the exact exchange potential on the hydrogen atom
(GP93)~\cite{GP93}. We solve this equation completely as an inverse
problem and use the resulting function $w(z)$---the enhancement factor it
defines, $F_{\rm GP93}$ (hereafter the GP93 factor)---as the single
nonempirical input of the design. The consequences come in three layers.

(i) \emph{Hydrogenic exactness} (Sec.~\ref{sec:theory1}). A GGA built
on the GP93 factor satisfies $\vx=-\vH$ on the hydrogen $1s$ density by
construction (hydrogen-exactness theorem), and scale covariance then
yields the exact eigenvalue $\varepsilon_{1s}=-Z^2/2$ for every
hydrogenic $1s$ ion (eigenvalue theorem). The GP93 factor itself contains
no empirical parameters; the switching parameters introduced later are
calibrated (Sec.~\ref{sec:fmix}).

(ii) \emph{Asymptotics and Rydberg-like bound states}
(Sec.~\ref{sec:theory2}). In the large-gradient regime the GP93 factor is
unbounded, with the hydrogen-exact growth $F=O[s(\ln s)^{2/3}]$---distinct
from the AK13 $s\ln s$ form---and on the hydrogen $1s$ density it endows
the exchange potential with the exact $-1/r$ tail by construction. An
integrated functional in which this tail-carrying $\Fmix$ is switched off
in many-electron regions, using the meta-GGA indicator
$\alpha=(\tau-\tau^W)/\tau^{\rm unif}$ of one-electron character,
produces Rydberg-like series of bound virtual states.

(iii) \emph{Predicted domain of applicability}
(Sec.~\ref{sec:theory3}). The switch activates the GP93 term only where
$\alpha\approx0$, i.e., where the density is described by a single
orbital. The effect therefore reaches only systems whose highest occupied
molecular orbital (HOMO) is a one-electron-like $s$ shell. Numerically,
Rydberg-like bound states are recovered (4--8 levels) in all five such
systems---all-electron H and He, and Li, Na, and K reduced to a single
valence electron by large-core pseudopotentials---and not in the
$p$-shell atoms Ne and Ar, exactly as the design predicts. That
pseudopotentials render the valence electron one-electron-like plays an
essential role here, which suggests a natural affinity with real-space
pseudopotential implementations.

The claims of this paper are limited to one-electron properties (Rydberg
states, asymptotics, hydrogenic levels, SIE). We make no claim about
thermochemistry (atomization energies): hydrogenic exactness pins the
absolute energy of the H atom to its exact value, which differs from the
error-compensated value that PBE was optimized to produce, so the two
cannot be reconciled within a semilocal form
(Sec.~\ref{sec:disc}). This explicitly delimited scope follows the
precedent of the LB94~\cite{LB94} and mBJ~\cite{mBJ} lineage
(potential-oriented functionals that do not target thermochemistry).

\section{Solution of the GP93 equation and hydrogenic exactness}\label{sec:theory1}

For a spin-unpolarized density we write the GGA exchange energy as
\begin{align}
\Ex[n]&=\Ax\int \rd^3r\; n^{4/3}\, F(s),\qquad
\Ax=-\tfrac34\big(\tfrac3\pi\big)^{1/3},
\nonumber\\
s&=\frac{|\nabla n|}{2(3\pi^2)^{1/3}n^{4/3}},
\label{eq:gga}
\end{align}
with spin polarization restored by the spin-scaling relation. On the
hydrogen $1s$ density the exact-exchange condition reads
\begin{equation}
\Ex[n_{\text{1el}}]=-E_{\mathrm H}[n_{\text{1el}}],\qquad
\vx=\frac{\delta\Ex}{\delta n}=-\vH .
\label{eq:1e}
\end{equation}
For hydrogen, $n_{\mathrm H}=\pi^{-1}\ee^{-2r}$,
$\vH(r)=r^{-1}[1-\ee^{-2r}(1+r)]$, and
$\vx^{\rm exact}\to-1/r$.

In the GP93 variables $x=|\nabla n|/n^{4/3}$ and
$z=3\ln[x/(2\pi^{1/3})]$, with the semilocal factor $g(x)=-x\,w(z)$, the
enhancement factor is
\begin{align}
\Ax F(s)&=g(x)=-x\,w(z),
\nonumber\\
F(s)&=-\frac{x\,w(z)}{\Ax}=\frac{8\pi}{3}\,s\,w(z),
\label{eq:Fw}
\end{align}
with $z=3\ln s+\ln 6\pi$. On the hydrogen density $x=2\pi^{1/3}\ee^{2r/3}$, so that
\begin{equation}
z=2r\;(\ge0),
\label{eq:z2r}
\end{equation}
i.e., the physical domain is restricted to $z\ge0$. Imposing
Eq.~\eqref{eq:1e}, $w$ obeys the inhomogeneous Kummer-type equation
\begin{equation}
z\,w''+(2-z)\,w'+\tfrac23\,w=\tfrac16\big[2-\ee^{-z}(z+2)\big],
\label{eq:ode}
\end{equation}
whose homogeneous part has parameters $a=-2/3$, $b=2$ with fundamental
solutions $M(a,b,z)$ (Kummer) and $U(a,b,z)$ (Tricomi).

Variation of constants gives
\begin{align}
w(z)&=\frac12-\frac{\Gamma(1/3)}{4}
\big[\,U(z)\,I_M(z)+M(z)\,I_U(z)\,\big],
\label{eq:wexact}\\
I_M&=\int_0^z W\,M\,\rd t,\qquad
I_U=\int_z^\infty W\,U\,\rd t,\nonumber
\end{align}
with $W=t(t+2)\ee^{-2t}$. As $z\to\infty$,
$I_U=\ee^{-2z}z^{8/3}\Phi(1/z)$ with
$\Phi(t)=\tfrac12+\tfrac{10}{9}t-\tfrac{20}{81}t^2
-\tfrac{430}{2187}t^3+O(t^4)$
(endpoint Laplace expansion; Appendix~\ref{app:laplace}) is exponentially
small and $I_M\to C_M$ converges, so
\begin{equation}
w(z)=\frac12-\frac{\Gamma(1/3)}{4}\,C_M\,U(z)+O(\ee^{-z}),
\qquad U(z)\sim z^{2/3}.
\label{eq:wasym}
\end{equation}
This $z^{2/3}$ growth generates, through Eq.~\eqref{eq:Fw},
$F\sim s\ln s$ and hence $\vx\sim-1/r$ (Sec.~\ref{sec:theory2}).

Because $U$ has a branch cut on the negative real axis, its continuation
to $z<0$ is not unique. Equation~\eqref{eq:ode}, however, has a regular
singular point at $z=0$ (indices $\rho=0,-1$), and demanding analyticity
at $z=0$ singles out the regular branch. For the regular series
$w=\sum_n c_nz^n$ the recursion is
\begin{equation}
(m{+}1)(m{+}2)\,c_{m+1}+\big({-}m+\tfrac23\big)c_m=r_m,
\label{eq:recur}
\end{equation}
where $\tfrac16[2-\ee^{-z}(z{+}2)]=\sum_m r_mz^m$. The single free
constant $c_0=w(0)$ is fixed by matching to the $z>0$ representation
Eq.~\eqref{eq:wexact}, giving
\begin{equation}
c_0=w(0)=0.68928799
\end{equation}
(matching points $z_m\in\{2,3,4,5\}$; stable to eight digits). The GP93
enhancement factor $F$ is thereby uniquely determined on the whole domain
(Fig.~\ref{fig:wF}).

\begin{theorem}[Hydrogen-exactness theorem]\label{thm:window}
The GGA of Eq.~\eqref{eq:gga} built from $F$ of
Eqs.~\eqref{eq:Fw} and \eqref{eq:wexact} satisfies
$\vx(r)=-\vH(r)$ on the hydrogen density by construction. If the range where the GP93 term acts
is restricted to a finite radius by the outer switch of
Sec.~\ref{sec:fmix}, exactness is preserved within that radius.
\end{theorem}
\begin{proof}
Equation~\eqref{eq:ode} is the identity derived from the requirement
$\vx[n_{\mathrm H}]=-\vH$; it holds as long as $F$ is related to its
solution $w$ by Eq.~\eqref{eq:Fw}.
\end{proof}

\begin{theorem}[Scale covariance: hydrogenic eigenvalue
theorem]\label{thm:scale}
If $F$ depends only on scale-invariant descriptors ($s$ and, below,
$\alpha$ and $q$), then under uniform scaling
$n_\kappa(\mathbf r)=\kappa^3n(\kappa\mathbf r)$ one has
$\Ex[n_\kappa]=\kappa\Ex[n]$. Consequently, for the hydrogenic densities
$n_Z=(Z^3/\pi)\ee^{-2Zr}$ the exchange-only Kohn--Sham equation yields
$\varepsilon_{1s}=-Z^2/2$ exactly.
\end{theorem}
\begin{proof}
From $s[n_\kappa](\mathbf r)=s[n](\kappa\mathbf r)$ it follows that
$\int n_\kappa^{4/3}F=\kappa\int n^{4/3}F$; the statement follows from
consistency with the external potential $-Z/r$ under $r\to Zr$.
Numerically the eigenvalues agree at the mHa level for $Z=1$--$5$
(Table~\ref{tab:hydro}).
\end{proof}

\section{Asymptotic $-1/r$ and the mixed functional}\label{sec:theory2}

The GGA exchange potential involves $F$, $F'$, and $F''$. In the outer
region of an exponentially decaying density $n\sim\ee^{-2\beta r}$ one
has $s\propto\ee^{2\beta r/3}\to\infty$, so the asymptotics are governed
by the large-$s$ behavior of $F$.

\begin{proposition}[Large-gradient behavior]\label{thm:ak13}
If $F$ is bounded, $\vx$ vanishes exponentially in the tail of any
exponentially decaying density; an unbounded large-$s$ enhancement is
therefore necessary for any nonvanishing long-range exchange potential.
AK13~\cite{AK13} employs the logarithmic growth
\begin{equation}
F_{\rm AK13}(s)\sim B_1\,s\ln s\qquad(s\to\infty).
\label{eq:slns}
\end{equation}
The GP93 factor is likewise unbounded, but with a different,
hydrogen-exact structure: inserting Eq.~\eqref{eq:wasym},
$w(z)=\tfrac12-\tilde C\,z^{2/3}+O(z^{-1/3})+O(\ee^{-z})$ with
$\tilde C=\Gamma(1/3)C_M/4$, into Eq.~\eqref{eq:Fw} gives
\begin{equation}
F_{\rm GP93}(s)=\frac{4\pi}{3}\,s
-\frac{8\pi}{3}\,\tilde C\,s\,(3\ln s)^{2/3}+\cdots,
\label{eq:gp93asym}
\end{equation}
i.e., $F_{\rm GP93}=O[s(\ln s)^{2/3}]$, in agreement with the
$O[x(\ln x)^{2/3}]$ growth noted by Gill and Pople~\cite{GP93}. On the
hydrogen $1s$ density this growth is not merely an asymptotic ansatz:
the GP93 equation itself ensures $\vx=-\vH$ and hence the Coulombic
$-1/r$ tail by construction (Theorem~\ref{thm:window}).
\end{proposition}

\begin{remark}[Asymptotics--stability trade-off]\label{rem:tradeoff}
The unbounded growth of the enhancement factor competes with many-electron SCF
stability and energy boundedness in the very same large-$s$ region.
Bounding $F$ (reverting to a PBE-like bounded form) restores stability but destroys
Eq.~\eqref{eq:coulombtail}. The two cannot coexist within a semilocal
form; this dictates the switch design of Sec.~\ref{sec:theory3} and the
scope discussed in Sec.~\ref{sec:disc}.
\end{remark}

\label{sec:fmix}%
The small-$s$ (GE2/PBE), $z\approx0$ (Frobenius series), and large-$z$
(divergent branch) regimes are joined by an inner switch $\eta(z)$ and an
outer switch
\begin{equation}
\chi(s)=1-\exp[-(s/s_0)^p],\qquad p\ge4,
\label{eq:chi}
\end{equation}
yielding $\Fmix=[1-\chi]\Fpbe+\chi F_{\mathrm{GP93}}$. The property
$\chi=O(s^p)$ ($p\ge4$) protects the second-order gradient expansion,
while on the hydrogen density $s\propto\ee^{2r/3}$ implies that $1-\chi$
vanishes doubly exponentially in $r$, so the tail is dominated by the
divergent branch. The parameters $(p,s_0)$ of the outer switch are not an
empirical fit: they are calibrated against the very purpose of the
functional---the recovery of bound Rydberg states and the agreement of
$-\varepsilon_{\rm HOMO}$ with experimental ionization energies. We
minimize the mean absolute error
$\mathrm{MAE}=\tfrac15\sum|{-}\varepsilon_{\rm HOMO}-I_{\rm exp}|$ over
the five systems with one-electron-like $s$-shell HOMOs (H, He, Li, Na,
K), within the region where all systems converge self-consistently
(Fig.~\ref{fig:calib}). For $p\gtrsim8$ the results saturate and the $p$
dependence is negligible. Decreasing $s_0$ improves the MAE, but for
$s_0\lesssim0.15$ the SCF of the alkali valence electron diverges. The
optimum over the stable domain is
\begin{equation}
(p^\star,s_0^\star)=(8,\;0.16),
\qquad \mathrm{MAE}=0.80\ \text{eV},
\end{equation}
which we use throughout. At this $s_0$ the exchange-only SCF density of
hydrogen reproduces the exact density closely
($\int_{r<2}|n-n_{\rm H}|\,\rd\mathbf r\simeq10^{-3}$) with
$\varepsilon_{1s}=-0.4998$; the Rydberg-oriented calibration is thus
compatible with hydrogenic exactness.

\begin{figure}[t]\centering
\includegraphics[width=\columnwidth]{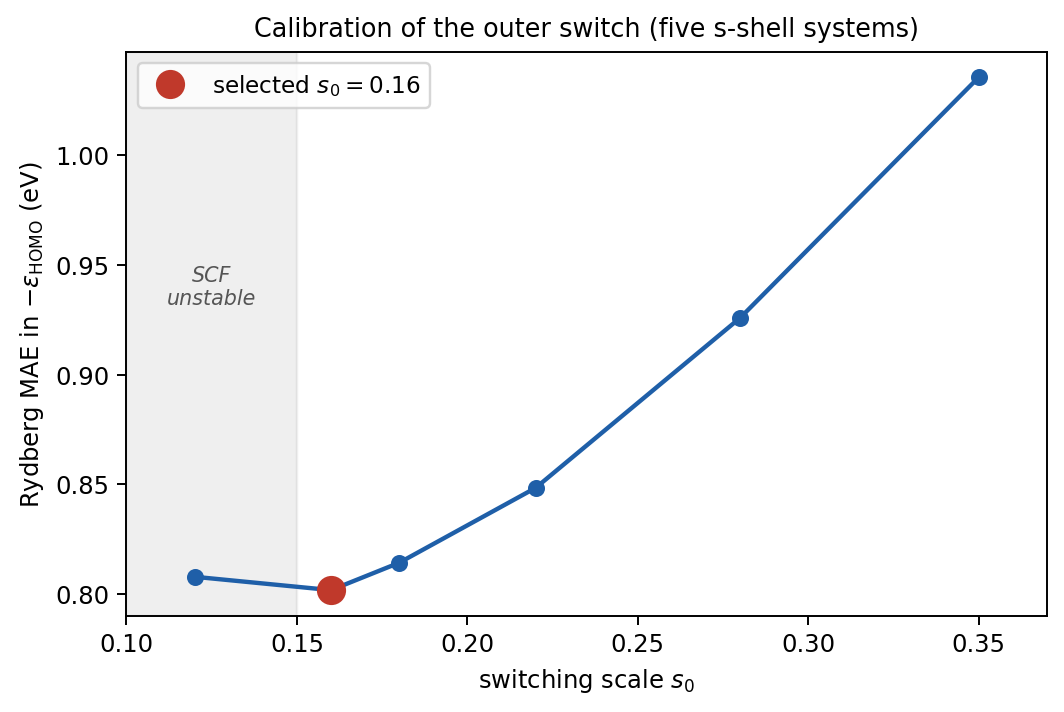}
\caption{Mean absolute error of the Rydberg series versus the
outer-switch scale $s_0$ (five $s$-shell systems: H, He, Li, Na, K). The
minimum occurs at $s_0=0.16$; for $s_0\lesssim0.15$ (shaded) the SCF of
the alkali valence electron diverges. The exponent is fixed to $p=8$,
beyond which the results saturate.}
\label{fig:calib}
\end{figure}

\section{Many-electron switching and the integrated functional}\label{sec:theory3}

From $\tau=\tfrac12\sum_i|\nabla\psi_i|^2$, $\tau^W=|\nabla n|^2/8n$, and
$\tau^{\mathrm{unif}}=\tfrac3{10}(3\pi^2)^{2/3}n^{5/3}$ we define
\begin{equation}
\alpha=\frac{\tau-\tau^W}{\tau^{\mathrm{unif}}}\ge0,
\qquad
q=\frac{\nabla^2n}{4(3\pi^2)^{2/3}n^{5/3}} .
\label{eq:alphaq}
\end{equation}
In any single-orbital region $\tau=\tau^W$, hence $\alpha\equiv0$:
$\alpha$ is an exact detector of one-electron character.

\begin{remark}[Separation limit of finite descriptor
sets]\label{thm:nogo}
For any finite set of scale-invariant descriptors $\{s,q,\dots,p_K\}$ one
can construct mimicking exponential densities that agree with hydrogen up
to order $K$ as $z\to\infty$. Separation is complete only in the limit
$K\to\infty$; in particular, with $(s,q)$ alone parts of heavy-atom cores
are degenerate with hydrogen in the high-density, small-$s$ regime, and
the $n^{4/3}$ weight then defeats the separation of exchange energies.
\end{remark}

The variables $\alpha$ and $q$ carry complementary local information:
$\alpha$ distinguishes the number of orbitals (single vs.\ many) but is
blind to the shape of a single orbital (Be-type: non-hydrogenic valence
with $\alpha\equiv0$), whereas $q$ captures the curvature of the density
but is insensitive to orbital count. The pair $(s,\alpha,q)$ therefore
provides a more discriminating local description than $(s,\alpha)$
alone; in the SCAN~\cite{SCAN} variable space $(s,\alpha)$ the Be-type
case cannot be separated.

The integrated functional is
\begin{equation}
\tilde F(s,\alpha)=\Fpbe(s)+G(\alpha)\,\big[\Fmix(s)-\Fpbe(s)\big],
\label{eq:integrated}
\end{equation}
with $G=\exp[-(\alpha/w_\alpha)^2]$ and $w_\alpha=0.10$. Since $G(0)=1$
and $G'(0)=0$, the switch does not perturb $\vx$ in one-electron regions.
Two constructive properties follow: (i) \emph{one-electron limit}---for
$\alpha\equiv0$ one has $\tilde F=\Fmix$, so the GP93 term and its
asymptotics are fully active (inheriting
Theorems~\ref{thm:window} and \ref{thm:ak13}); (ii) \emph{suppression in many-electron
regions}---where $\alpha>0$, $\tilde F\to\Fpbe$, preventing the
divergence of $\Fmix$ from destroying many-electron energies (indeed,
the SCF of the HF molecule, which diverges with bare $\Fmix$, converges
with the integrated form).

A prediction of the domain of applicability follows immediately: the
asymptotic improvement reaches the HOMO only if the HOMO itself is
one-electron-like ($\alpha\approx0$), i.e., only in systems whose
outermost shell is occupied by a single $ns$ electron. For $p$-shell
HOMOs ($\alpha>0$) the switch suppresses the GP93 term and no effect appears.
This prediction is verified numerically in Sec.~\ref{sec:results}.

\subsection{$q$-resolved switching toward molecules and solids}\label{sec:qgate}

The $\alpha$ switch of Eq.~\eqref{eq:integrated} distinguishes the type of
atomic HOMO ($s$ vs.\ $p$ shell) but can be insufficient in molecules and
solids. At the center of a covalent bond, the overlap of two atomic
orbitals produces a low-density, large-$s$ region with $\alpha\approx0$,
where the divergent branch of $\Fmix$ causes numerical instability. This
region cannot be distinguished from the atomic tail (the source of the
Rydberg states, where the GP93 term \emph{should} act) by $\alpha$ alone.

The discriminating key is the reduced Laplacian of
Eq.~\eqref{eq:alphaq}. That $q$ carries information independent of
$\alpha$ can be verified explicitly on an exponentially decaying atomic
tail. For $n(r)\sim\ee^{-2\kappa r}$ the reduced gradient
$s=|\nabla n|/[2(3\pi^2)^{1/3}n^{4/3}]$ grows with $r$, and
$\nabla^2 n=n''+\tfrac{2}{r}n'=(4\kappa^2-\tfrac{4\kappa}{r})n>0$, so
\begin{equation}
q(r)\xrightarrow{r\to\infty}
\frac{\kappa^2}{(3\pi^2)^{2/3}n^{2/3}}\longrightarrow+\infty,
\label{eq:qtail}
\end{equation}
i.e., $q$ diverges to large positive values (the density is convex).
At a covalent bond center, by contrast, the overlapping densities form a
saddle with $\nabla^2 n\le0$, hence $q\lesssim0$. This sign difference
arises even when $\alpha=0$ in both regions, because $\alpha$ measures
the one-electron character of the \emph{kinetic energy} while $q$
measures the \emph{curvature of the density}---independent pieces of
information. Numerically, atomic tails show $q\sim10$--$10^2$ while bond
centers show $q\sim-0.5$--$1$ (Sec.~\ref{sec:results},
Table~\ref{tab:qgate_mol}). Exploiting this independence we define the
triple-switched form
\begin{equation}
\tilde F(s,\alpha,q)=\Fpbe+D(q)\,G(\alpha)\,[\Fmix-\Fpbe],
\label{eq:triple}
\end{equation}
with $D(q)=\tfrac12[1+\tanh\{(q-q_c)/w_q\}]$, $q_c=1$, and $w_q=0.6$.
The switch $D$ has the two limits
\begin{equation}
D(q)\xrightarrow{q\gg q_c}1,\qquad D(q)\xrightarrow{q\ll q_c}0.
\label{eq:Dlim}
\end{equation}
By Eq.~\eqref{eq:qtail}, atomic tails have $q\to+\infty$ and thus
$D\to1$: the $\alpha$-switched form of Eq.~\eqref{eq:integrated} is
recovered and the GP93 term remains active. At covalent bond centers
$q\lesssim0<q_c$ gives $D\to0$, hence $\tilde F\to\Fpbe$ and the
divergent branch of $\Fmix$ is removed. At polar and ionic bond centers
$q$ may be positive, but $\alpha>0$ already suppresses the GP93 term through
$G(\alpha)\to0$. The product $D(q)G(\alpha)$ closes whenever either
factor vanishes,
\begin{equation}
D(q)G(\alpha)\to0\quad\text{if}\quad
\underbrace{q\lesssim q_c}_{\text{covalent}}\ \text{or}\
\underbrace{\alpha>0}_{\text{$p$ shell, polar, many-electron}},
\end{equation}
suppressing the divergence in the representative bonding environments
tested here (fixed-density tests for ten molecules in
Sec.~\ref{sec:results}).

Importantly, the $q$ switch does not damage the atomic $-1/r$ tail at all.
This follows from the separation of roles between $q$ and $s$: an atomic
tail has $s\to\infty$ \emph{and} $q\to+\infty$, and by
Eq.~\eqref{eq:Dlim} the switch is saturated there,
$\partial D/\partial q_c\to0$, so that varying $q_c$ leaves the atomic
$-1/r$ tail unchanged; only near bond centers ($q\sim q_c$) does
$\partial D/\partial q_c\ne0$, where suppression is tuned. That is,
\begin{equation}
\frac{\partial(-r\vx)}{\partial q_c}\bigg|_{\text{tail}}\!\approx0,
\qquad
\frac{\partial\,[D G]}{\partial q_c}\bigg|_{\text{bond}}\!\ne0,
\end{equation}
i.e., in the cases tested the atomic tail is insensitive to $q_c$ while
the bond-center suppression is tuned by it (Sec.~\ref{sec:results}). This is in sharp contrast to an $s$-dependent
switch, which would also suppress the large-$s$ atomic tail and destroy the Rydberg
physics; the sign structure of Eq.~\eqref{eq:qtail} makes $q$ essential.

Including $q$ (the Laplacian) in the energy introduces the functional
derivative $v_{\rm lapl}=\partial e_x/\partial(\nabla^2 n)$, which---as
for SCAN-type meta-GGAs---is numerically delicate in a self-consistent
potential and unstable in Gaussian-basis implementations. In this paper
we therefore demonstrate the discriminating power of the triple switch
non-self-consistently (on fixed densities) and leave the self-consistent
implementation to real-space methods, where $v_{\rm lapl}$ can be handled
stably on a grid (Sec.~\ref{sec:disc}).

In all-electron calculations the cores of heavy atoms pass through the
enhancement region of the GP93 factor and over-enhance the exchange (the degeneracy noted in
Remark~\ref{thm:nogo}); for Be-type systems the SCF diverges.
Large-core pseudopotentials remove the core, thereby (i) eliminating the
source of over-enhancement and (ii) rendering the alkali valence
literally one-electron ($ns^1$), which makes $\alpha\approx0$ exact.
Large-core pseudopotentials thus play a constructive role in the present
framework by isolating a one-electron-like valence density.

The exchange is combined with an independent correlation functional
(as for AK13).
\begin{proposition}[LYP compatibility]\label{prop:lyp}
LYP correlation~\cite{LYP} vanishes for any one-electron density
($E_c=0$) and therefore preserves the hydrogenic exactness of the
present exchange (Theorems~\ref{thm:window} and \ref{thm:scale}).
Numerically, the H-atom eigenvalue with the integrated exchange plus LYP
is $\varepsilon_{1s}=-0.4994$, essentially identical to the
exchange-only value $-0.4999$, whereas PBE correlation shifts it to
$-0.5085$ (Table~\ref{tab:corr}).
\end{proposition}

\section{Computational methods}\label{sec:method}

All calculations use the custom-functional interface
(\texttt{define\_xc\_}) of PySCF~\cite{PySCF} (version 2.13), with
correlation taken from libxc~\cite{libxc} (version 7.0; PBE or LYP). The
exchange is implemented with analytic channel derivatives
$v_\rho,v_\sigma,v_\tau$, with $F$ and $F'$ represented by cubic splines
in $\ln s$ over $s\in[10^{-5},10^{4}]$ and clipped outside this range;
densities are floored at $10^{-250}$ in the evaluation of reduced
variables. Default PySCF quadrature grids of level 5--7 are used, with
SCF convergence thresholds of $10^{-8}$--$10^{-9}$~Ha; where convergence
is difficult, damping (0.2--0.3) and level shifts (0.2~Ha) are applied in
a first stage and removed for the final iterations. Spin-restricted (RKS)
and spin-unpolarized densities are treated through the spin-scaling
relation; open-shell systems use unrestricted (UKS) calculations. Rydberg-like bound states are identified in
exchange-only calculations, where the presence or absence of the
potential tail is reflected directly, by counting negative Kohn--Sham
virtual eigenvalues ($\varepsilon_{\rm virt}<0$). Diffuse basis
sets~\cite{Dunning,augBasis} are essential to represent the tail
(aug-cc-pV5Z for H and He; aug-cc-pVQZ/TZ for Ne and Ar). Because a
finite Gaussian basis provides only a discretized representation of the
near-continuum spectrum, the precise \emph{count} of such states can
depend on the diffuseness of the basis; the qualitative contrast reported
below (zero bound states for the reference functionals versus a series of
bound states for the present one, stable over the calibration range
$s_0=0.13$--$0.28$) is the basis-robust content of our claim. Alkali
atoms are reduced to a single valence electron by large-core
pseudopotentials (CRENBL for Li~\cite{CRENBL}; SBKJC for Na and
K~\cite{SBKJC,SBKJCK}). Experimental ionization energies are taken from
the NIST Atomic Spectra Database~\cite{NIST}. The $\alpha$ switch converges with standard DIIS
because $v_\tau$ enters as a Hermitian kinetic-energy-like operator;
where convergence is difficult, staged damping and level shifting are
used.

The radial exchange potential $-r\,\vx(r)$ of Fig.~\ref{fig:vx} is
evaluated on densities solved self-consistently \emph{with} correlation
for each functional, taking the exchange part only,
\begin{equation}
\vx(\mathbf r)=v_\rho-\big[f'(r)+\tfrac2r f(r)\big],
\qquad f(r)=2\,v_\sigma\,\partial_r n,
\label{eq:vxproj}
\end{equation}
with the divergence term constructed by finite differences (the same
procedure as in a real-space code). On the hydrogen density Theorem~\ref{thm:window}
gives $v_x=-\vH$ exactly, so $-r\,\vx=1-\ee^{-2r}(1+r)\to1$ is
fixed analytically. When the self-consistent $\Fmix$ density is expanded
in Gaussians, however, the outer region deviates from the exact density
and $-r\,\vx$ departs from the exact curve. This is a mathematical
property of the basis: hydrogen is a one-electron system, so
Hartree--Fock equals exact exchange and the exact density
$n=\ee^{-2r}/\pi$ is known; Gaussian HF densities reproduce it to within
$0.3\%$ relative error in the intermediate region ($r\lesssim3.5$) at the
aug-cc-pV5Z level, but in the far region ($r\gtrsim4$,
$n\lesssim10^{-5}$) a finite Gaussian set cannot represent the
exponential tail $\ee^{-2r}$ (even though the total energy is converged
to $-0.499995$~Ha at V5Z). Figure~\ref{fig:vx} is therefore restricted
to the inner regions where the basis is reliable. A quantitative
evaluation of the potential tail for all systems, and the
self-consistent implementation of the $q$-switched form (which requires a
stable treatment of $v_{\rm lapl}$), are better suited to real-space
pseudopotential methods; in this paper the primary evidence is the bound
Rydberg series, which is the physical consequence of the asymptotic
tail.

\section{Results}\label{sec:results}

\subsection{Hydrogenic exactness}

Table~\ref{tab:hydro} lists the exchange-only eigenvalues of the
hydrogenic series. As required by Theorem~\ref{thm:scale},
$\varepsilon_{1s}=-Z^2/2$ is reproduced at the mHa level; no empirical
parameter enters the GP93 factor itself.

\begin{table}[b]\centering
\caption{Hydrogenic series (exchange only): numerical demonstration of
Theorem~\ref{thm:scale}. Energies in hartree.}
\label{tab:hydro}
\begin{ruledtabular}
\begin{tabular}{lrr}
System & $-\varepsilon_{1s}$ (calc.) & $Z^2/2$ \\
\hline
H & 0.49972 & 0.5 \\
He$^+$ & 1.99989 & 2 \\
Li$^{2+}$ & 4.50144 & 4.5 \\
\end{tabular}
\end{ruledtabular}
\end{table}

\begin{figure*}[t]\centering
\includegraphics[width=0.95\textwidth]{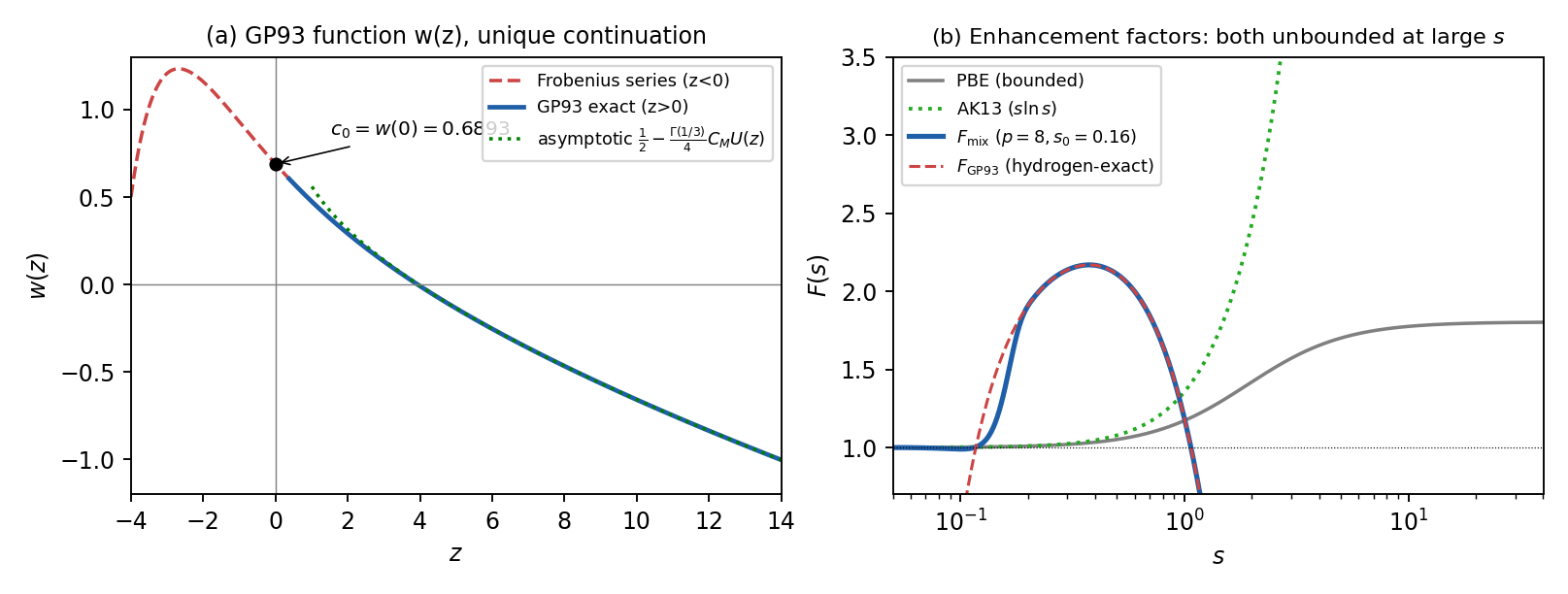}
\caption{(a) The GP93 auxiliary function $w(z)$: the Frobenius series
($z<0$, dashed), the exact $z>0$ representation of
Eq.~\eqref{eq:wexact} (solid), and the asymptotic form of
Eq.~\eqref{eq:wasym} (dotted) join uniquely at $z=0$
($c_0=0.68929$). (b) Enhancement factors $F(s)$: PBE is bounded and therefore lacks the
hydrogenic tail. AK13 has an $s\ln s$-type growth, whereas the GP93
factor carried by $\Fmix$ ($s_0{=}0.16$) has the hydrogen-exact
large-gradient form $O[s(\ln s)^{2/3}]$ of
Eq.~\eqref{eq:gp93asym} (Proposition~\ref{thm:ak13}). Both are
unbounded, but their asymptotic structures and design principles differ;
$\Fmix$ is constructed from hydrogen exactness
(Theorem~\ref{thm:window}) as its nonempirical input.}
\label{fig:wF}
\end{figure*}

\subsection{Rydberg-like bound states (main result)}

\begin{figure*}[t]\centering
\includegraphics[width=0.62\textwidth]{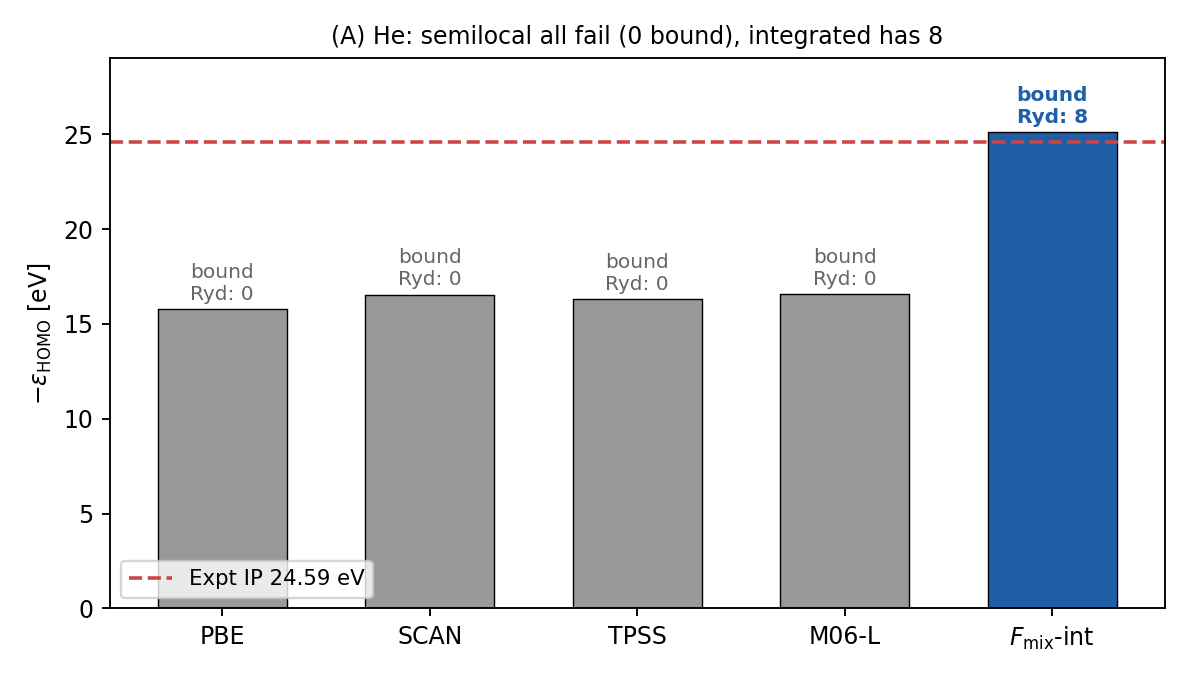}\\[1.2mm]
\includegraphics[width=0.48\textwidth]{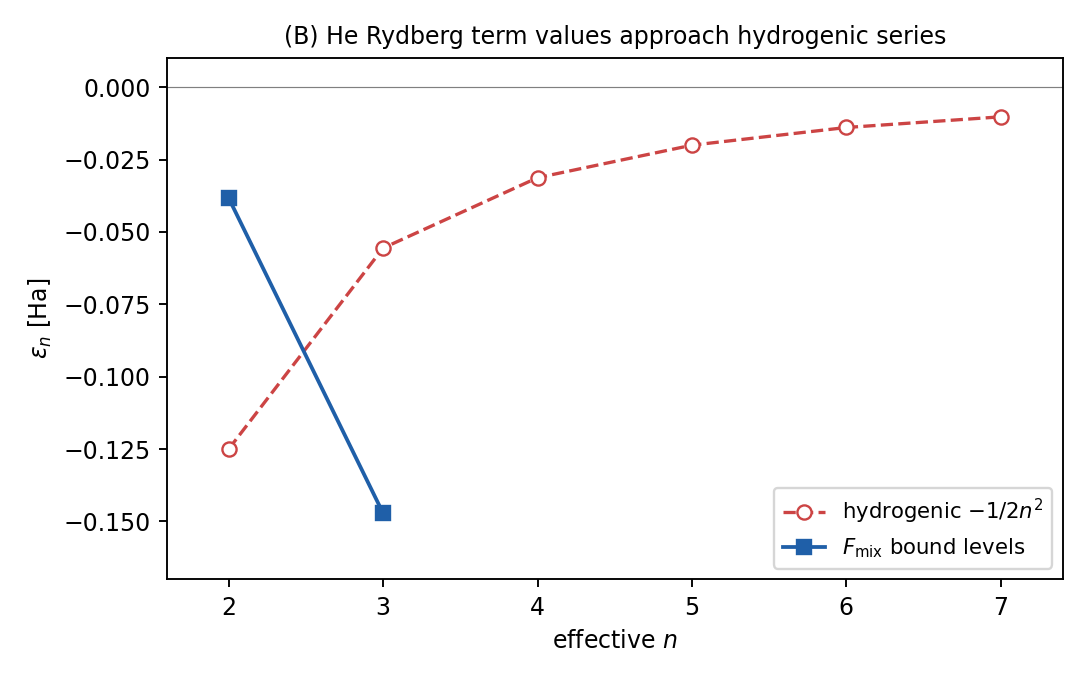}\hfill
\includegraphics[width=0.48\textwidth]{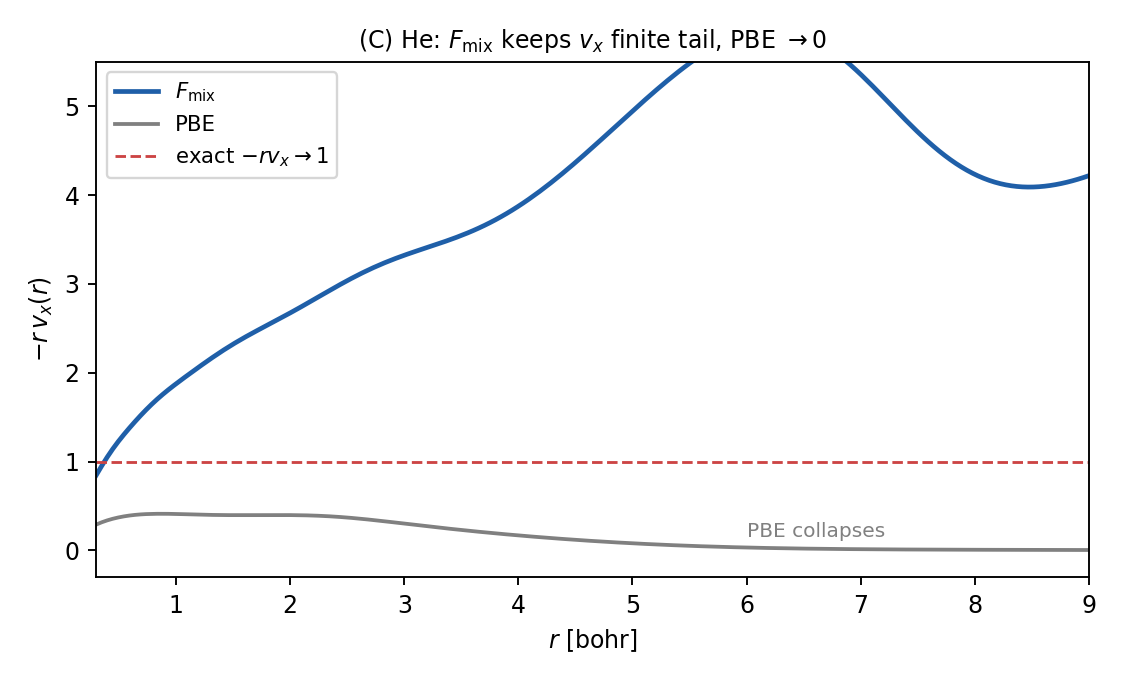}
\caption{Helium atom. (A) $-\varepsilon_{\rm HOMO}$ and number of bound
Rydberg orbitals: PBE, SCAN, TPSS, and M06-L all bind zero, while the
integrated $\Fmix$ binds eight with $-\varepsilon_{\rm HOMO}=25.1$~eV
(experiment 24.59). (B) The two bound levels ($n_{\rm eff}=1.84$ and $3.63$). $\Fmix$ binds
the lowest two Rydberg-like levels; shallower
high-$n$ levels remain unbound owing to the residual of the potential
tail (a systematic error of the hydrogen-referenced calibration). Even
so, the existence of bound levels is a qualitative difference from the
semilocal and hybrid functionals, all of which bind none. (C)
$-r\vx(r)$: $\Fmix$ retains a finite tail, PBE collapses to zero. That
$\Fmix$ does not reach the exact value 1 asymptotically is the same
systematic calibration error responsible for the slight overestimate of
the IP (disclosed as part of the scope).}
\label{fig:ryd}
\end{figure*}

\begin{figure*}[t]\centering
\includegraphics[width=0.9\textwidth]{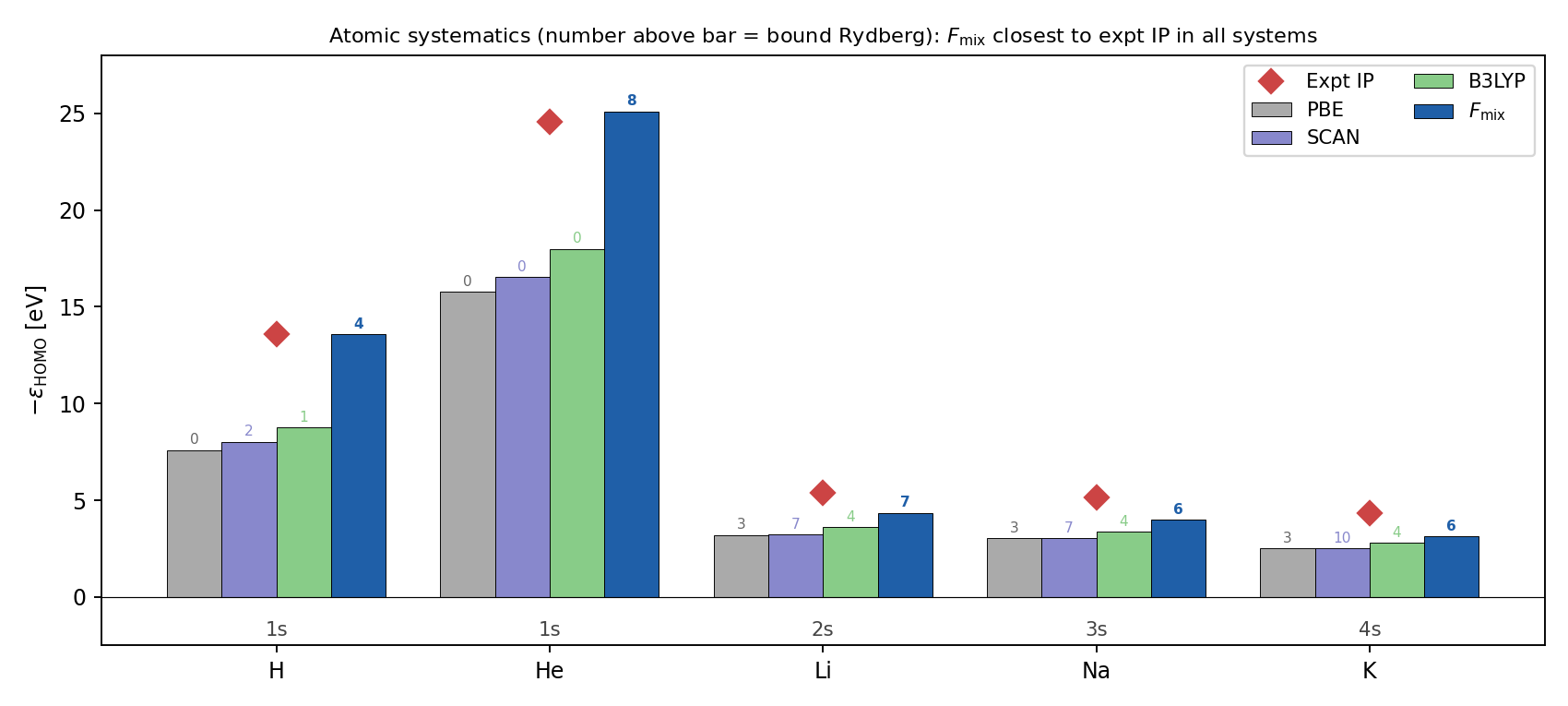}
\caption{Atomic systematics, all-functional comparison (number above
each bar = bound Rydberg orbitals). $-\varepsilon_{\rm HOMO}$ of
PBE~\cite{PBE}, SCAN~\cite{SCAN}, B3LYP~\cite{B3LYP}, and $\Fmix$ (dark
blue), with experimental ionization potentials (red diamonds). In all
five systems (H, He, Li, Na, K; one-electron-like $s$-shell HOMOs)
$\Fmix$ is closest to experiment. The mean absolute HOMO--IP errors are
PBE 4.19, SCAN 3.95, B3LYP 3.30, and $\Fmix$ 0.79~eV, four to five
times better than the other functionals. For H and He, $\Fmix$ also
stands out in the number of bound levels (4--8 versus 0--2). Li uses
CRENBL, Na and K use SBKJC (one valence electron).}
\label{fig:rydatom}
\end{figure*}

For the He atom (Fig.~\ref{fig:ryd}), PBE, SCAN, TPSS, and M06-L all
fail to bind any Rydberg state ($-\varepsilon_{\rm HOMO}\approx16$~eV,
more than 8~eV below the experimental 24.59), while the integrated
$\Fmix$ alone yields eight bound orbitals and
$-\varepsilon_{\rm HOMO}=25.1$~eV. The independent bound levels are two, at $-0.147$ and $-0.038$~Ha; their
effective quantum numbers $n_{\rm eff}=1/\sqrt{2|\varepsilon|}=1.84$ and
$3.63$ place them close to a hydrogenic progression. The count of eight
orbitals includes the angular-momentum partners that accompany these two
levels in the diffuse basis; shallower high-$n$ levels are not bound
because of the residual of the potential tail.
Even so, the existence of bound levels is a qualitative difference from
the semilocal and hybrid functionals, which bind none
(Table~\ref{tab:allfunc}; Fig.~\ref{fig:ryd}B). The radial $-r\vx$
[Fig.~\ref{fig:ryd}(C)] visualizes the origin directly at the level of
the potential: $\Fmix$ retains a finite tail while PBE collapses to
$4\times10^{-3}$ by $r=9$~bohr.

\begin{table*}[t]\centering
\caption{Rydberg-like bound states, all-functional comparison. Each
column pair gives $-\varepsilon_{\rm H}\equiv-\varepsilon_{\rm HOMO}$
[eV] and the number of bound virtual orbitals (including
angular-momentum partners). Experimental ionization
energies~\cite{NIST} in the last row. See the text for discussion.}
\label{tab:allfunc}
\begin{ruledtabular}
\begin{tabular}{llrrrrrr}
 & & \multicolumn{2}{c}{He ($1s$)} & \multicolumn{2}{c}{Ne ($2p$)} &
\multicolumn{2}{c}{Li$^{\rm ECP}$ ($2s$)} \\
Functional & Type & $-\varepsilon_{\rm H}$ & bound &
$-\varepsilon_{\rm H}$ & bound & $-\varepsilon_{\rm H}$ & bound \\
\hline
PBE & GGA & 15.76 & 0 & 13.35 & 0 & 3.21 & 3 \\
BLYP & GGA & 15.92 & 0 & 13.37 & 0 & 3.00 & 4 \\
SCAN & meta-GGA & 16.54 & 0 & 14.00 & 0 & 3.24 & 7 \\
B3LYP & hybrid (20\%) & 18.00 & 0 & 15.65 & 0 & 3.61 & 4 \\
PBE0 & hybrid (25\%) & 18.22 & 0 & 16.00 & 0 & 3.80 & 3 \\
SCAN0 & hybrid (25\%) & 18.77 & 0 & 16.48 & 0 & --- & --- \\
$\Fmix$ integrated & meta-GGA & 25.12 & 8 & 11.85 & 0 & 4.35 & 7 \\
Experiment & --- & 24.59 & --- & 21.56 & --- & 5.39 & --- \\
\end{tabular}
\end{ruledtabular}
\end{table*}

The He/Ne contrast in Table~\ref{tab:allfunc} shows that the effect is
not specific to He but follows from the design principle
(one-electron-like $s$-shell HOMO). Li illustrates a different regime:
with a single valence electron and small $Z_{\rm eff}$, other functionals
also bind Rydberg-like states (PBE 3, SCAN 7), but $\Fmix$ is closest to
the experimental ionization energy (4.35 versus 5.39~eV) with among the
largest bound counts---the advantage shifts from ``the only functional
that binds'' (He) to ``the best description'' (Li).

The atomic systematics (Fig.~\ref{fig:rydatom}, Table~\ref{tab:atoms})
constitute the central result of this paper. In all five systems with
one-electron-like $s$-shell HOMOs---H ($1s$), He ($1s$), and Li ($2s$),
Na ($3s$), K ($4s$) reduced to one valence electron by large-core
pseudopotentials---Rydberg-like series of bound states (4--8 levels) are recovered and
$-\varepsilon_{\rm HOMO}$ is consistently closer to the experimental IP
than PBE. In the $p$-shell atoms Ne ($2p$) and Ar ($3p$) the HOMO region
has $\alpha>0$, the switch suppresses the GP93 term, and the bound count remains
zero. This ``$s$-shell-only'' selectivity is consistent with the
intended design domain: the $\alpha$ switch selects one-electron-like
HOMOs (Sec.~\ref{sec:theory3}), so that the domain of applicability
follows from the design itself.

\begin{table}[b]\centering
\caption{Atomic systematics of bound Rydberg states, all-functional
comparison. Each cell gives $-\varepsilon_{\rm HOMO}$ [eV] / number of
bound Rydberg orbitals. In all five systems $\Fmix$ brings
$-\varepsilon_{\rm HOMO}$ closest to the experimental IP (compared not
only with PBE but also with SCAN and B3LYP). For H and He, $\Fmix$ also
stands out in bound counts (4--8 versus 0--2). For the valence-only
alkalis (Li--K) other functionals also bind states, but the
$-\varepsilon_{\rm HOMO}$ of $\Fmix$ is consistently the best. Li:
CRENBL; Na, K: SBKJC (one valence electron each). $s_0{=}0.16$
(Rydberg-accuracy calibration); the MAE of HOMO--IP over the five
$s$-shell systems is 0.79~eV. For the $p$-shell references (Ne, Ar) see
the text and Table~\ref{tab:allfunc}.}
\label{tab:atoms}
\begin{ruledtabular}
\begin{tabular}{llrrrrr}
Atom & HOMO & Expt.\ IP & PBE & SCAN & B3LYP & $\Fmix$ \\
\hline
H & $1s$ & 13.60 & 7.59/0 & 8.02/2 & 8.77/1 & 13.59/4 \\
He & $1s$ & 24.59 & 15.76/0 & 16.54/0 & 18.00/0 & 25.12/8 \\
Li$^{\rm ECP}$ & $2s$ & 5.39 & 3.21/3 & 3.24/7 & 3.61/4 & 4.35/7 \\
Na$^{\rm ECP}$ & $3s$ & 5.14 & 3.02/3 & 3.03/7 & 3.38/4 & 3.99/6 \\
K$^{\rm ECP}$ & $4s$ & 4.34 & 2.52/3 & 2.50/10 & 2.80/4 & 3.12/6 \\
\end{tabular}
\end{ruledtabular}
\end{table}

\subsection{Many-electron stability and H$_2$}

The $\alpha$ switch of Eq.~\eqref{eq:integrated} stabilizes molecular SCF
calculations---which diverge with bare $\Fmix$ in the low-density,
large-$s$ bonding region---wherever the density localizes near the
nuclei. In the binding region of H$_2$ (Fig.~\ref{fig:h2}), combined
with LYP correlation, the equilibrium bond length is
$R_e\simeq0.70$~\AA\ (FCI/experiment 0.74), while the well depth of
3.6~eV is shallower than the FCI value of 4.7~eV---the generic
limitation of an exchange-focused semilocal functional
(thermochemistry is outside the claims of this paper). On the
dissociation side ($R\gtrsim1.5$~\AA) the broken-symmetry UKS diverges
in the low-density internuclear region and cannot be converged: this is
precisely the region where the divergent branch of $\Fmix$ is not fully
suppressed by the $\alpha$ switch alone, motivating the $q$ switch of
Sec.~\ref{sec:qgate} and the real-space implementation.

\begin{figure}[t]\centering
\includegraphics[width=\columnwidth]{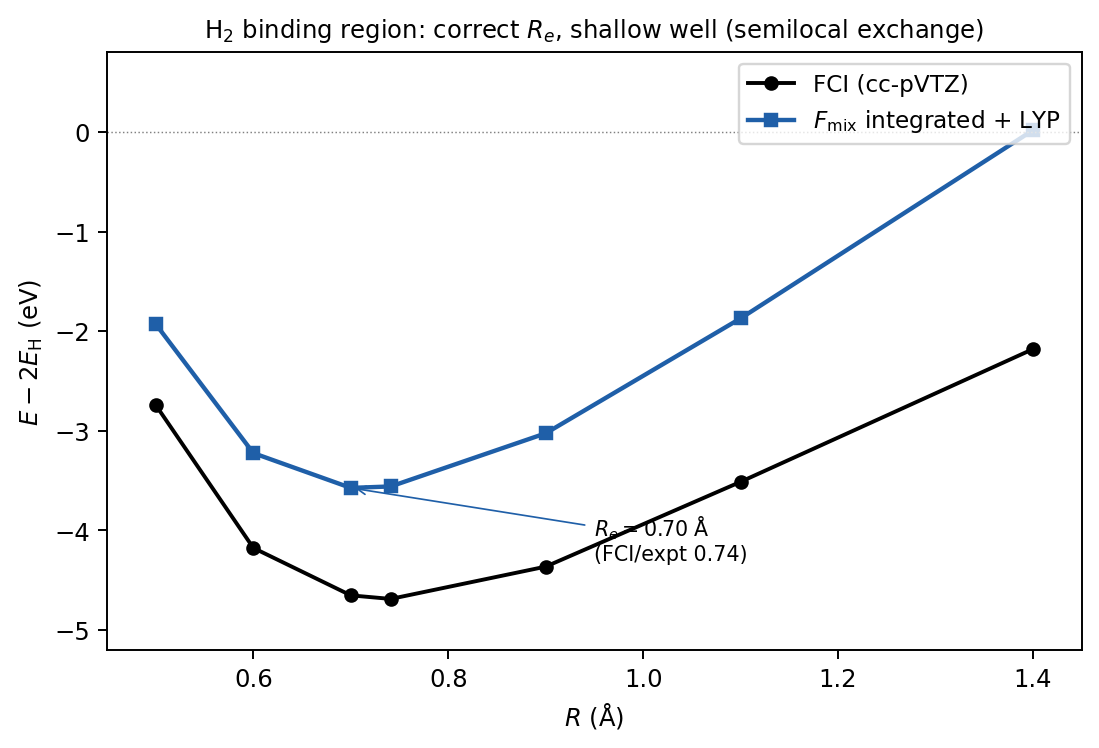}
\caption{Binding-region potential curve of H$_2$ with the integrated
functional (LYP correlation; broken-symmetry UKS). The bond length
$R_e\simeq0.70$~\AA\ is close to FCI and experiment (0.74~\AA); the well
depth of 3.6~eV is shallower than the FCI value of 4.7~eV. Shown is the
binding region in which stable self-consistent solutions exist.}
\label{fig:h2}
\end{figure}

\begin{figure*}[t]\centering
\includegraphics[width=0.95\textwidth]{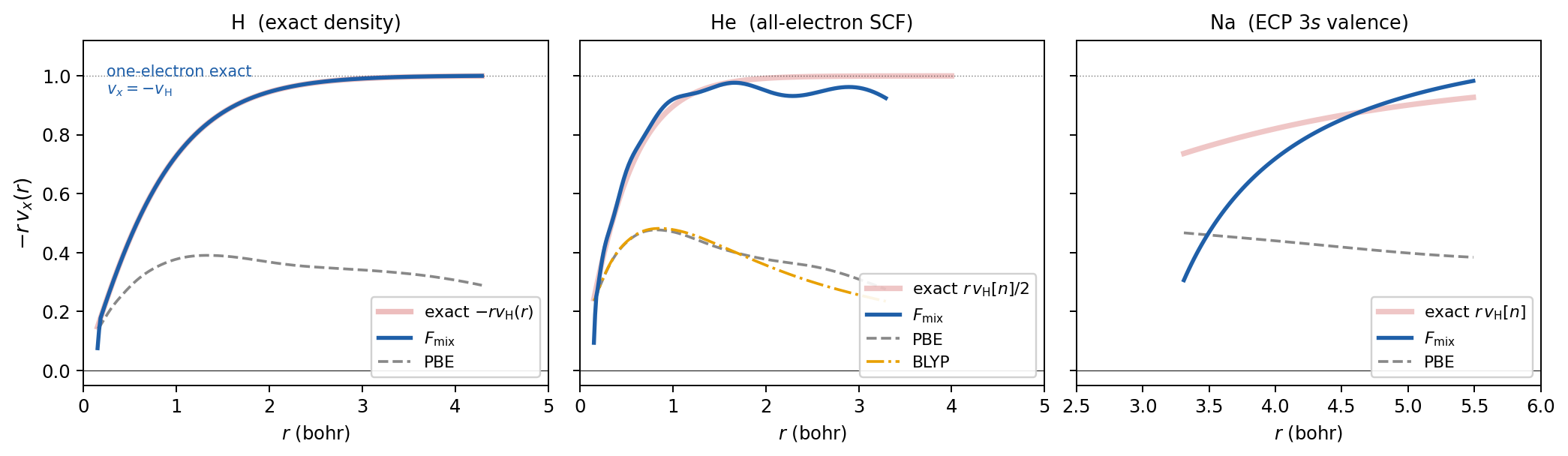}
\caption{Radial exchange potential $-r\,\vx(r)$ evaluated on
self-consistent densities obtained with correlation, taking the exchange
part only (finite-difference construction). In each panel the red curve
is the \emph{exact} exchange reference computed from the same density:
for H (exact density) $v_x=-v_{\rm H}$ (Theorem~\ref{thm:window}); for
He, a closed-shell two-electron system, the exact relation
$v_x=-v_{\rm H}/2$ holds; for Na, a one-electron valence system,
$v_x=-v_{\rm H}$ holds, with $r\,v_{\rm H}[n]$ (halved for He)
constructed by radial integration. For H, $\Fmix$ coincides with the
exact curve. For He, $\Fmix$ maintains $-r\,\vx\simeq1$ while PBE and
BLYP~\cite{B88,LYP} collapse to 0.3--0.5. For Na, $\Fmix$ approaches the
exact reference outside the valence-density maximum, reaching 0.98 at
$r=5.5$ (exact 0.97). For He and Na the curves are restricted to the
range where the Gaussian basis is reliable. Only $\Fmix$ recovers the
$-1/r$ tail, which is the origin of the Rydberg-like bound states.}
\label{fig:vx}
\end{figure*}

\subsection{Systematic tests of the triple switch ($q$ resolution)}\label{sec:qgate_results}

That the $q$ switch of Eq.~\eqref{eq:triple} suppresses the divergence
where the $\alpha$ switch is insufficient in molecules is verified
systematically in non-self-consistent evaluations (on fixed PBE or HF
densities). A self-consistent $q$ switch is unstable in Gaussian bases
owing to the $v_{\rm lapl}$ problem (Sec.~\ref{sec:disc}), but the
discriminating power itself is unambiguous on fixed densities. The
quantity monitored is whether the triple switch $D(q)G(\alpha)$ opens
($\to1$) or closes ($\to0$) in each region.

For atoms, Table~\ref{tab:qgate_atoms} lists $\alpha$ near the HOMO of
the second-period elements together with whether the exchange potential
retains its $-1/r$ tail (GP93 term active). The $s$-shell HOMOs of H and He
have $\alpha=0$; with $q$ large and positive in the tail ($10$--$10^2$),
the GP93 term acts and $-r\vx\simeq0.9$ is maintained. The $p$-shell HOMOs
from B to Ne have $\alpha=0.7$--$40$, so $G(\alpha)\approx0$ suppresses the GP93 term
completely and $\vx$ reverts to PBE. This confirms the prediction
of Sec.~\ref{sec:theory3} (the GP93 term acts only for $s$ shells) across
the entire second period.

\begin{table}[b]\centering
\caption{Triple-switch behavior near atomic HOMOs
(non-self-consistent). $\alpha$ is the value in the HOMO region;
the last column indicates whether $\vx$ keeps the $-1/r$ tail: yes for
$s$ shells ($\alpha{=}0$), no for $p$ shells ($\alpha{>}0$).}
\label{tab:qgate_atoms}
\begin{ruledtabular}
\begin{tabular}{llrl}
Atom & HOMO & $\alpha$ (HOMO region) & $-1/r$ tail \\
\hline
H  & $1s$ & 0.00 & yes ($-r\vx{\simeq}0.94$) \\
He & $1s$ & 0.00 & yes ($-r\vx{\simeq}0.88$) \\
B  & $2p$ & 1.08 & no (reverts to PBE) \\
C  & $2p$ & 1.04 & no \\
N  & $2p$ & 1.41 & no \\
O  & $2p$ & 0.59 & no \\
F  & $2p$ & 1.58 & no \\
Ne & $2p$ & 5.56 & no \\
\end{tabular}
\end{ruledtabular}
\end{table}

For molecules, Table~\ref{tab:qgate_mol} lists the triple-switch value at
the bond center near equilibrium for ten molecules representing diverse
bonding types. In covalent bonds (H$_2^+$, H$_2$, BH, CO, N$_2$) the
bond-center $q$ is small or negative ($-0.5$--$0$) and $D(q)$ closes; in
polar and ionic bonds (LiH, HF, LiF) $\alpha>0$ and $G(\alpha)$ closes.
In total, the bond-center switch is below 0.1 for nine of the ten
molecules, suppressing the divergent branch of $\Fmix$. The only
borderline case is He$_2^+$ (both $q$ and $\alpha$ intermediate; switch
0.11). The two switches thus act complementarily and suppress the
divergence in the representative bonding environments tested here.

\begin{table}[b]\centering
\caption{Triple-switch values at molecular bond centers
(non-self-consistent, near equilibrium). Values close to 0 mean the
GP93 term is suppressed. Covalent bonds are switched off via
$q$; polar and ionic bonds via $\alpha$.}
\label{tab:qgate_mol}
\begin{ruledtabular}
\begin{tabular}{lrrrl}
Molecule & Switch & $q$ & $\alpha$ & Closing mechanism \\
\hline
H$_2^+$ & 0.008 & $-0.45$ & 0.00 & $q$ (covalent) \\
H$_2$   & 0.012 & $-0.31$ & 0.00 & $q$ (covalent) \\
BH      & 0.005 & $-0.36$ & 0.09 & $q$ (covalent) \\
CO      & 0.008 & $-0.31$ & 0.07 & $q$ (covalent) \\
N$_2$   & 0.000 & $-0.08$ & 0.55 & $q{+}\alpha$ \\
LiH     & 0.000 & $+2.10$ & 6.06 & $\alpha$ (ionic) \\
HF      & 0.000 & $+0.19$ & 0.62 & $\alpha$ (polar) \\
LiF     & 0.000 & $+5.50$ & 13.1 & $\alpha$ (ionic) \\
HeH$^+$ & 0.064 & $+0.20$ & 0.00 & $q$ \\
He$_2^+$& 0.114 & $+0.63$ & 0.08 & borderline \\
\end{tabular}
\end{ruledtabular}
\end{table}

That these two mechanisms leave the atomic $-1/r$ tail intact is verified by
scanning $q_c$. Moving $q_c$ from 0 to 1 leaves the He tail fixed at
$-r\vx=0.90$, while the H$_2^+$ bond-center switch decreases monotonically
from 0.33 to 0.017. Because the atomic tail ($q\sim30$) and the bond
center ($q\sim-0.2$) are completely separated in $q$, one can be tuned
without sacrificing the other. This is fundamentally different from an
$s$-dependent switch, which would also suppress the large-$s$ atomic tail and
destroy the Rydberg states, and it establishes that discrimination by
the reduced Laplacian $q$ is essential for the extension to molecules
and solids.

\subsection{Effect of the correlation functional}

Table~\ref{tab:corr} shows the correlation dependence. Because LYP
vanishes for one-electron systems, the H-atom eigenvalue
$\varepsilon_{1s}=-0.4994$ essentially coincides with the exchange-only
value ($-0.4999$), preserving hydrogenic exactness
(Proposition~\ref{prop:lyp}); PBE correlation injects a finite
correlation into the one-electron system and shifts the eigenvalue to
$-0.5085$. Bond length and well depth are essentially identical for the
two correlations, so LYP, which preserves hydrogenic exactness, is the
recommended companion of the present exchange.

\begin{table}[b]\centering
\caption{Correlation functionals compared in the H$_2$ binding region
(Proposition~\ref{prop:lyp}). LYP vanishes for one-electron systems and
keeps the H-atom eigenvalue at the exchange-only value ($-0.4999$),
whereas PBE correlation injects finite correlation. Bond length and well
depth are equivalent; the well being shallower than FCI is the generic
limitation of an exchange-focused semilocal functional (thermochemistry
is not claimed).}
\label{tab:corr}
\begin{ruledtabular}
\begin{tabular}{lrrr}
 & $+$LYP & $+$PBEc & Reference \\
\hline
$R_e$ [\AA] & 0.70 & 0.70 & FCI/expt.\ 0.74 \\
Well depth [eV] & 3.6 & 3.6 & FCI 4.7 \\
H atom $\varepsilon_{1s}$ & $-0.4994$ & $-0.5085$ & x-only $-0.4999$ \\
\end{tabular}
\end{ruledtabular}
\end{table}

\section{Orbital-independent triple switch}\label{sec:orbfree}

The indicator $\alpha=(\tau-\tau^W)/\tau^{\rm unif}$ in the triple switch
of Eq.~\eqref{eq:triple} depends on the Kohn--Sham orbitals through the
kinetic-energy density $\tau$ and therefore requires a generalized
Kohn--Sham (gKS) solver. In this section we show that $\tau$ can be
replaced by an orbital-free kinetic-energy-density functional (KEDF),
turning $\alpha$ into a semilocal density functional $\alpha(s,q)$. The
entire triple switch then becomes a pure density functional: its
exchange potential is an ordinary local (multiplicative) potential
common to all orbitals, so that no gKS solver is required in any
Kohn--Sham code. Its only remaining nonstandard ingredient is the
Laplacian $\nabla^2 n$, whose potential term $v_{\rm lapl}$ is unstable
in Gaussian bases but is evaluated naturally on real-space grids
(Sec.~\ref{sec:disc}); real-space pseudopotential codes are therefore
the natural target.

Writing the kinetic-energy density through an enhancement factor,
$\tau=\tau^{\rm unif}F_s$ with
$\tau^{\rm unif}=\tfrac{3}{10}(3\pi^2)^{2/3}n^{5/3}$,
$\tau^W=\tau^{\rm unif}F_s^W$, and $F_s^W=\tfrac53 p$ ($p=s^2$), we have
\begin{equation}
\alpha=\frac{\tau-\tau^W}{\tau^{\rm unif}}=F_s-F_s^W .
\label{eq:alpha_Fs}
\end{equation}
Making $\alpha$ orbital-free thus reduces to approximating the true
(orbital) enhancement factor $F_s$ by a semilocal density functional
$F_s(s,q)$. The exact iso-orbital condition $\alpha=0$ corresponds to
$F_s=F_s^W$, i.e., to saturation of the lower bound $\tau=\tau^W$.

Naive choices fail. The second-order gradient expansion, and
Pauli--Gaussian-type KEDFs~\cite{PGSL}
($F_s^P=e^{-\mu p}+\beta q^2$ with $F_s^P\equiv F_s-F_s^W$), are
designed for the accuracy of the \emph{integrated} kinetic energy
$T_s=\int\tau$ and do not guarantee the \emph{pointwise} bound
$F_s\ge F_s^W$. Indeed, the $\beta q^2$ term diverges in the atomic
tail, where $q\to+\infty$ [Eq.~\eqref{eq:qtail}], so
$\alpha=e^{-\mu p}+\beta q^2\to\infty\ne0$, violating one-electron
exactness.

The Laplacian-level KEDF of Perdew and Constantin
(PC07)~\cite{PC07} is essential here precisely because it imposes the
pointwise bound $\tau^W\le\tau$ ($F_s\ge F_s^W$) constructively. The
starting point is the fourth-order gradient expansion
\begin{equation}
F_s^{\rm GE4}=1+\tfrac{5}{27}p+\tfrac{20}{9}q+\Delta,
\quad \Delta=\tfrac{8}{81}q^2-\tfrac19 pq+\tfrac{8}{243}p^2,
\label{eq:ge4}
\end{equation}
which, however, diverges as $\Delta\sim\tfrac{8}{81}q^2$ for
$q\to\infty$. Its bounded modification,
\begin{equation}
F_s^{\rm GE4M}=\frac{F_s^{\rm GE4}}{\sqrt{1+[\Delta/(1+\tfrac53 p)]^2}}
\xrightarrow{|q|\to\infty}1+\tfrac53 p=1+F_s^W,
\label{eq:ge4m}
\end{equation}
saturates onto $F_s^W$ as $|q|\to\infty$ (the denominator grows like
$|\Delta|$). Finally, $F_s^{\rm GE4M}$ is smoothly interpolated to the
lower bound $F_s^W$,
\begin{equation}
F_s^{\rm PC}=F_s^W+(F_s^{\rm GE4M}-F_s^W)\,f_{ab}(F_s^{\rm GE4M}-F_s^W),
\label{eq:fspc}
\end{equation}
\begin{equation}
f_{ab}(z)=
\begin{cases}
\left[\dfrac{1+\ee^{a/(a-z)}}{\ee^{a/z}+\ee^{a/(a-z)}}\right]^{b}, &
0<z<a,\\[2ex]
1, & z\ge a,
\end{cases}
\label{eq:fab}
\end{equation}
with $a=0.5389$ and $b=3$. Because $f_{ab}(0)=0$, one has
$F_s^{\rm PC}=F_s^W$ wherever $F_s^{\rm GE4M}=F_s^W$ (one-electron
regions), so that by Eq.~\eqref{eq:alpha_Fs}
$\alpha^{\rm PC}=(F_s^{\rm GE4M}-F_s^W)f_{ab}(\cdot)\ge0$: the lower
bound is guaranteed. This construction reproduces the orbital indicator
at covalent bond centers ($\alpha\approx0$), in $p$-shell atoms
($\alpha>0$), and in the inner one-electron region ($\alpha=0$)
(Table~\ref{tab:orbfree}).

In the outermost atomic tail, however, the bounded form
Eq.~\eqref{eq:ge4m} does not return to the von Weizs\"acker limit. For an
exponentially decaying density both $p$ and $q$ grow without bound, with
$q/p=1-1/(\kappa r)\to1$; in this regime
$\Delta\simeq\tfrac{5}{243}p^2>0$ dominates, so that
$F_s^{\rm GE4M}\to1+F_s^W$ rather than $F_s^W$, and consequently
$\alpha^{\rm PC}\to1$ instead of the exact value $0$. At the finite radii
where Gaussian densities are meaningful ($\kappa r\simeq3$--$6$,
$q/p\simeq0.7$--$0.85$) the same saturation already produces a positive
residual. As a pragmatic regularization of the density-only
approximation---not a parameter-free consequence of the GP93
construction---we impose a large-$p$ damping factor,
\begin{equation}
\alpha^{\rm PC}(s,q)=\big(F_s^{\rm GE4M}-F_s^W\big)f_{ab}(\cdot)\;
\ee^{-(p/p_c)^2},\quad p_c=5,
\label{eq:alphapc}
\end{equation}
forcing $\alpha^{\rm PC}\to0$ as $p\to\infty$ (outermost tail), which
restores $\alpha=0$ (GP93 term active) in the asymptotic region that determines
the Rydberg levels.

\begin{table}[b]\centering
\caption{Orbital-independent $\alpha^{\rm PC}(s,q)$ versus the orbital
indicator $\alpha^{\rm KS}$ (non-self-consistent, on PBE/HF densities).
Covalent bond centers, the $p$ shell, and one-electron tails are
reproduced.}
\label{tab:orbfree}
\begin{ruledtabular}
\begin{tabular}{lrr}
Region & $\alpha^{\rm KS}$ & $\alpha^{\rm PC}(s,q)$ \\
\hline
H tail ($r\ge2$) & 0.00 & 0.00 \\
He tail ($r\ge3$) & 0.00 & 0.00 \\
Ne ($2p$ shell) & 0.44--0.73 & 0.68--0.76 \\
H$_2$ bond center & 0.00 & 0.00 \\
N$_2$ bond center & 0.55 & 0.545 \\
CO bond center & 0.07 & 0.003 \\
LiH bond center & 6.06 & 5.58 \\
\end{tabular}
\end{ruledtabular}
\end{table}

The integrated functional built on the orbital-independent
$\alpha^{\rm PC}$ yields, in the outer tail of He ($r\ge3$, the region
that determines the Rydberg levels), exactly the same exchange-potential
tail as the orbital version ($-r\vx\simeq0.90$): the Rydberg asymptotics
are recovered \emph{without orbitals}. In the intermediate region
($r\simeq1.5$--$2$, $p\simeq2$--$8$), by contrast, the second-derivative
structure of the density ($p$, $q$, $q/p$) of one-electron-like He
happens to be close to that of $p$-shell Ne; local quantities alone
cannot fully decide the one-electron character there and a residual of
about 0.3 remains in $\alpha^{\rm PC}$. This residual mildly perturbs
$\vx$ near $r\simeq2$ but recovers for $r\ge3$, so the impact on HOMO
and Rydberg levels---determined in the outer region---is limited. This
is an intrinsic limitation of local discrimination by the density;
refinement with higher-order descriptors is left for future work.

The methodological consequence is significant: the one-electron
detection at the core of meta-GGA functionality can be approximated
\emph{as a density functional} by an orbital-free KEDF possessing the
$\tau^W\le\tau$ bound structure, so that the entire triple switch
requires no gKS solver---only a standard local exchange potential;
combined with the grid-friendly treatment of $\nabla^2 n$, this makes
self-consistent implementation in real-space Kohn--Sham codes feasible. This opens the path from the
non-self-consistent demonstrations of Sec.~\ref{sec:qgate_results} to
full self-consistency.

\section{Discussion}\label{sec:disc}

LB94 and mBJ construct potentials directly and lack variational
consistency. AK13 realized Eq.~\eqref{eq:slns} within the GGA energy
form, but its construction is independent of hydrogenic exactness. The
present framework differs in that a single nonempirical input---the GP93
hydrogen-exactness condition---produces both hydrogenic exactness
(Theorems~\ref{thm:window} and \ref{thm:scale}) and the unbounded large-gradient growth (Proposition~\ref{thm:ak13}) from the same function $w(z)$. The
contrast with SCAN is direct: SCAN refines the intermediate-range
description with $\alpha$ but its potential tail remains exponential
[Fig.~\ref{fig:ryd}(A)]; the present framework uses the same variable
$\alpha$ as the criterion for suppressing the GP93 term, thereby recovering
bound Rydberg states at comparable computational cost.

Two independent checks show that the bound Rydberg states are not an
artifact of the calibration. First, predictivity on one-electron cations
not used in the calibration: for He$^+$ ($Z{=}2$) and Li$^{2+}$
($Z{=}3$) the integrated $\Fmix$ gives $-\varepsilon_{1s}=54.44$ and
$122.50$~eV, in essentially exact agreement with $Z^2/2$ (54.42 and
122.45~eV), whereas PBE gives 42.0 and 103.7~eV. This is simultaneously
a numerical demonstration of Theorem~\ref{thm:scale} and evidence that
exact one-electron IPs are reproduced independently of the calibration
parameters. Second, robustness of the He Rydberg states against $s_0$:
over the range $s_0=0.13$--$0.28$, both
$-\varepsilon_{\rm HOMO}=25.12$~eV and the count of eight bound
orbitals are unchanged; $s_0$ only fine-tunes the alkali
$-\varepsilon_{\rm HOMO}$, and the qualitative binding in
one-electron-like HOMO systems does not depend on it. The bound Rydberg
series is therefore a consequence of the structural property
$v_x\to-1/r$, not an artifact of fitting.

Table~\ref{tab:allfunc} shows that even global hybrid functionals
containing partial exact exchange (B3LYP~\cite{B3LYP},
PBE0~\cite{PBE0}, SCAN0~\cite{SCAN0}; 20--25\% mixing) bind no He
Rydberg-like state in the present calculations. Global hybrids improve
the HOMO eigenvalue (here from about 16 to 18~eV) by introducing a
fraction of nonlocal exact exchange, but their long-range potential
contains only the corresponding fraction of the full Coulombic tail and
therefore does not reach the full $-1/r$ Kohn--Sham asymptote;
range-separated hybrids~\cite{LC,CAMB3LYP} are usually required when a
quantitatively correct long-range potential is the primary target. The
present framework, by contrast, recovers the $-1/r$ behavior within a
purely semilocal form, without exact-exchange admixture and its
attendant nonlocal cost.

The limitations of the framework are structural, not incidental.
(i) \emph{Thermochemistry}: hydrogenic exactness pins the absolute
energy of the H atom, which differs from the error-compensated PBE
value, so atomization energies involving H are necessarily shifted.
(ii) \emph{$p$-shell HOMOs}: since the one-electron detection is based
on $\alpha$, the asymptotic improvement does not reach many-electron
HOMOs (Ne, Ar). (iii) \emph{Quantitative IPs of heavier systems}: owing
to the systematic error of the hydrogen-referenced calibration,
$-\varepsilon_{\rm HOMO}$ is highly accurate for He but not
quantitative in general. (iv) \emph{Response properties}: because of
the density sensitivity of the switch boundaries, higher derivatives such
as polarizabilities are outside the scope. With these limitations made
explicit, the claim---that Rydberg-like bound states, asymptotics, and
hydrogenic levels of $s$-shell one-electron-like HOMO systems are
recovered---with a nonempirical tail-generating ingredient and a small
number of calibrated switching parameters, at semilocal cost---where all
existing semilocal functionals including SCAN fail
qualitatively---is consistent with every number presented.

The Li--K systematics (Table~\ref{tab:atoms}) showed that large-core
pseudopotentials activate the GP93 term by rendering the valence electron
one-electron-like. This suggests that the functional operates naturally
in real-space pseudopotential implementations. The fully variational
implementation of the $q$-switched form of Eq.~\eqref{eq:triple} is also
expected to realize its full value on real-space grids, where
$v_{\rm lapl}$ can be assembled naturally.

In this paper the discriminating power of the triple switch was
demonstrated non-self-consistently (Sec.~\ref{sec:qgate_results}). When
a self-consistent implementation is attempted in a Gaussian basis,
dropping the $v_{\rm lapl}$ term associated with $q$ causes the SCF to
diverge through the inconsistency between energy and potential, while
including it correctly requires a nonlocal $V_{\rm xc}$ assembly
involving second derivatives of the basis functions. This is the same
class of difficulty that Laplacian-level meta-GGAs face, and it is
impractical in Gaussian bases. On a real-space grid, by contrast,
$\nabla^2 n$ and its functional derivative are evaluated naturally by
finite differences and $v_{\rm lapl}$ can be handled stably. The
self-consistent implementation of the triple switch therefore belongs to
real-space methods, where the discriminating power established
here---no trade-off; atomic tails preserved while molecular
divergences are suppressed---should be realized in full.

Finally, in extended systems the $-1/r$ correction is switched off in the smooth
bulk interior through the $\alpha$ and $q$ switches, so the functional
reduces there to PBE (neither benefit nor breakdown). The correction acts
in low-density, one-electron-like regions---the vacuum side of surfaces
and the exterior of low-dimensional systems---exactly where an
image-potential-like $-1/r$ asymptote is physically required and
systematically missed by semilocal functionals. In preliminary
periodic-boundary calculations a simplified switched variant converged
stably for solid He over a range of lattice constants, confirming
numerical robustness under periodic boundary conditions. Once the
real-space implementation of the triple switch is realized, applications
are anticipated to problems where the vacuum-side asymptote affects
adsorption energetics, such as hydrogen adsorption on metal and
semiconductor surfaces---an application domain complementary to bulk
thermochemistry, which this functional does not target.

\section{Conclusions}\label{sec:concl}

We solved the GP93 hydrogen-exactness condition completely as an inverse
problem and, from the GP93 factor defined by its solution, constructed a
semilocal exchange ingredient that is exact on the hydrogen $1s$ density
and on its coordinate-scaled hydrogenic family, with no empirical
parameter entering the ingredient itself. Embedded in a switched
meta-GGA form with a small number of calibrated switching parameters,
this ingredient produces Rydberg-like series of bound virtual states in
all of H, He, Li, Na, and K---systems whose highest occupied orbital is
a one-electron-like $s$ shell---where existing semilocal functionals
from PBE to SCAN bind none in the same calculations. The effective domain
($s$ shells yes, $p$ shells no) is predicted by the design and confirmed
numerically; pseudopotentials that render the valence electron
one-electron-like are the key that activates the GP93 term. As an extension
toward molecules and solids we introduced a triple switch that adds the
reduced Laplacian $q$ as a third indicator, and showed on fixed
densities---for the second-period atoms and ten molecules---that it
distinguishes atomic tails from covalent bond centers and suppresses the
molecular divergence. Finally, replacing the kinetic-energy density
$\tau$ by an orbital-free functional with the $\tau^W\le\tau$ bound
(Perdew--Constantin type) renders the entire switch a density-only
functional and reproduces the Rydberg asymptotics without orbitals,
removing the need for a generalized Kohn--Sham solver, while the
remaining Laplacian ingredient is naturally handled on real-space
grids. A fully
self-consistent real-space pseudopotential implementation, systematic
basis- and grid-convergence tests of the bound-state counts, and
applications to problems where the vacuum-side asymptote matters---such
as hydrogen adsorption on surfaces---remain necessary next steps.

\begin{acknowledgments}
This work received no funding or external support. During the research, AI assistants---Claude (Claude Opus 4.8, Anthropic) and GPT-5.5 (OpenAI)---were used interactively to help explore analytic manipulations. The author independently verified all derivations and results and bears sole responsibility for the content.
\end{acknowledgments}

\appendix

\section{Endpoint Laplace expansion of the tail integral}\label{app:laplace}

Substituting
$U(t)\sim t^{2/3}(1-\tfrac{10}{9t}-\tfrac{10}{81t^2}+\cdots)$ into
$I_U=\int_z^\infty W\,U\,\rd t$ and applying termwise
\begin{equation}
\int_z^\infty\ee^{-2t}t^\nu\rd t=\tfrac12\ee^{-2z}z^\nu
\Big[1+\tfrac{\nu}{2z}+\tfrac{\nu(\nu-1)}{(2z)^2}+\cdots\Big]
\end{equation}
yields
$\Phi(t)=\tfrac12+\tfrac{10}9t-\tfrac{20}{81}t^2-\tfrac{430}{2187}t^3
+O(t^4)$. The relative error of the three-/four-term truncation is
below $2.3\times10^{-2}$/$2.1\times10^{-3}$ for $z\in[2,20]$.

\section{Frobenius coefficients and determination of $c_0$}\label{app:frob}

Decomposing the recursion Eq.~\eqref{eq:recur} as $c_n=A_nc_0+B_n$
($A_0=1$, $B_0=0$) gives
$A_{m+1}=\frac{(m-2/3)A_m}{(m+1)(m+2)}$ and
$B_{m+1}=\frac{r_m+(m-2/3)B_m}{(m+1)(m+2)}$. Then
$c_0=[w_{\rm exact}(z_m)-\psi_N(z_m)]/\phi_N(z_m)$, with
$\phi_N=\sum A_nz^n$, $\psi_N=\sum B_nz^n$, and $N=60$, is stable to
eight digits for $z_m\in\{2,3,4,5\}$.

\section{Verification of the implementation}\label{app:verify}

The numerical correctness of the $\Fmix$ implementation was verified in
several independent layers. The analytic functional derivatives
$v_\rho$ and $v_\sigma$ agree with numerical differentiation to better
than $10^{-9}$. Scale covariance (Theorem~\ref{thm:scale}) is confirmed
on the hydrogenic series, with $\varepsilon_{1s}=-Z^2/2$ reproduced to
0.13~mHa for H, 0.10~mHa for He, and 1.6~mHa for Li. The exchange energy
from direct grid integration coincides with the value reported by PySCF
to machine precision ($2\times10^{-16}$). The SCF density is a true
energy minimum (energy increases under occupied--virtual rotations), and
three different initial densities converge to the same solution. These
checks strongly suggest that the reported trends are not implementation
artifacts.

\end{document}